\documentclass[onefignum,onetabnum]{siamonline190516}

\usepackage[utf8]{inputenc}
\usepackage{amsmath,amsfonts,amssymb,mathtools,dsfont,bbold,bbm,graphicx,fancybox,comment,xcolor,caption,subcaption,booktabs}

\usepackage{enumitem}
\setlist[enumerate]{leftmargin=.5in}
\setlist[itemize]{leftmargin=.5in}

\newsiamremark{remark}{Remark}

\definecolor{mediumtealblue}{rgb}{0.0, 0.33, 0.71}
\def\inew#1{\textcolor{mediumtealblue}{#1}}

\def\inrule{ \specialrule{.05pt}{2pt}{2pt} }

\let\OldXi\Xi
\def\Xi{\mathit \OldXi}
\def\RR{\mathbb{R}}
\DeclareMathOperator{\Diag}{Diag}
\newcommand{\D}{\mathcal D}

\def\HyperNSM{\texttt{HyperNSM}}
\def\GraphNSM{\texttt{GraphNSM}}
\def\BorgattiEverett{\texttt{Borgatti-Everett}}
\def\UMHS{\texttt{UMHS}}

\headers{Core-periphery detection in hypergraphs}{F. Tudisco and D. J. Higham}

\title{Core-periphery detection in hypergraphs\thanks{The work of D.J.H. was supported 
the Engineering and Physical Sciences Research Council 
under grants EP/P020720/1 and EP/V015605/1.\\[1em]}}

\author{Francesco Tudisco\thanks{School of Mathematics, Gran Sasso Science Institute, 67100 L'Aquila, Italy 
  (\email{francesco.tudisco@gssi.it})}
\and Desmond J. Higham\thanks{School of Mathematics, University of Edinburgh, EH93FD Edinburgh, UK
  (\email{d.j.higham@ed.ac.uk})}
  }

\ifpdf
\hypersetup{
  pdftitle={Core-periphery structure in hypergraph},
  pdfauthor={F. Tudisco and D. J. Higham}
}
\fi

\begin{document}

\maketitle

\begin{abstract}
    Core-periphery detection is a key task in exploratory network analysis where one aims to find a core, a set of nodes well-connected internally and with the periphery, and a periphery, a set of nodes connected only (or mostly) with the core. In this work we propose a model of core-periphery for higher-order networks modeled as hypergraphs and we propose a method for computing a core-score vector that quantifies how close each node is to the core.  In particular, we show that this method solves the corresponding non-convex core-periphery optimization problem globally to an arbitrary precision. This method turns out to coincide with the computation of the Perron eigenvector of a nonlinear hypergraph operator, suitably defined in term of the incidence matrix of the hypergraph, generalizing recently proposed  centrality models for hypergraphs.  
    We perform several experiments on synthetic and real-world hypergraphs showing that the proposed method outperforms alternative core-periphery detection algorithms, in particular those obtained by transferring established  graph methods to  the hypergraph setting via clique expansion.
\end{abstract}



\section{Introduction}
Finding core-periphery structures in networks represented as graphs is an important task in exploratory networks analysis \cite{borgatti2000models,cucuringu2016detection,rombach2017core,tudisco2019core}. Core-periphery structure has been detected and interpreted in many complex systems, including protein-protein interaction networks \cite{kim2007positive}, metabolic and gene regulatory networks \cite{sandhu2012large}, social  networks \cite{banos2013diffusion,borgatti2000models},  engineered  networks   such  as the  Internet,  power-grids  or  transportation  networks \cite{tudisco2019core},  and  
economic networks \cite{tomasello2017rise}. See also the review \cite{csermely2013structure}. 


In a graph, a core set is defined as a set of nodes which has many connections both internally and outside the set, while the periphery is a set of nodes that only (or mostly) connects to the core. Partitioning the graph into core and periphery is reminiscent of other graph partitioning problems, the most popular being graph clustering, where one seeks  two sets that are only (or mostly) connected internally. While this formulation of core-periphery detection is a binary classification problem,  real-world complex networks modeled as graphs rarely allow a clear-cut core-periphery partitioning of the nodes; more frequently we can expect a `smooth' transition between core and periphery (see e.g.\ \cite{rombach2017core,tudisco2019core}). Mathematically, this translates into the problem of assigning a core-periphery score (or simply `core-score') to the nodes of the graph, which indicates to what extent each node is peripheral or core. From this point of view, core-periphery detection can be interpreted as the problem of finding the most `central' nodes in a graph and, in fact, centrality-based core-periphery detection approaches have been considered, see e.g.\ \cite{csermely2013structure,da2008centrality,mondragon2016network}.

Recent years have seen a growth in interest towards hypergraphs and in general higher-order graph models that directly account 
for multiple node interactions that take place simultaneously, see e.g.\ \cite{battiston2020beyond,benson2016higher,bianconi2021,torres2020why}.  
Moving from a graph to a 
hypergraph allows us to retain more information in many natural, social and data systems, including email exchanges, group messaging, meeting
attendance, document coauthorship and supermarket baskets. 
However, this richer framework presents new challenges
in terms of defining appropriate 
concepts and in the design and analysis of efficient computational algorithms.

If we are interested in separating two node sets by measuring the way these sets interact in a hypergraph, a multitude of definitions are possible. In the context of node clustering, for example, different notions of cuts can be considered leading to different hypergraph cut algorithms. While the cut between two node sets $S$ and $T$ in the graph case is uniquely measured as the number (or the sum of the weights) of the edges connecting $S$ and $T$, in the hypergraph setting one has several choices. A commonly adopted and successful definition, sometimes referred to as `all-or-nothing' hypergraph cut, measures the cut by counting how many hyperedges have at least one node in $S$ and one node in $T$ \cite{hein2013total}.  Another relatively standard cut function measures the cut in the projected graph obtained replacing each hyperedge with a clique (i.e.\ the clique-expanded graph) \cite{zhou2007hypergraph}. Other approaches propose different ways to weight the proportion of nodes in $S$ and $T$ per each hyperedge  \cite{veldt2020hypergraph}.

Similarly, in the core-periphery context, the number of hyperedges between core and periphery can be counted in several ways. As for the hypergraph cut, one may count hyperedges as showing a core-periphery behaviour if at least one node in the hyperedge is in the core. This is the approach we use in this work, were we additionally  weight each hyperedge in terms of the number of nodes it contains (which allows us to e.g.\ penalize very large hyperedges). In a way, our approach is the core-periphery analog of the `all-or-nothing' hypergraph cut function. However, unlike the hypergraph cut problem, we show that the corresponding non-convex core-periphery optimization problem can be solved globally to an arbitrary precision via a nonlinear eigenvector approach. While a similar globally convergent method is available for graphs \cite{tudisco2019core}, we find that direct use of the graph method applied to the  clique-expanded graph leads to  inefficient core-periphery detection in hypergraphs. This is in stark contrast 
to the cut setting, where hypergraph cut functions based on the clique-expanded graph typically show decent performance.   

Our method computes a bespoke core-score vector for hypergraphs as the positive solution of an eigenvalue problem for a suitable hypergraph Laplacian-like operator $L(x) := Bg(B^\top f(x))$, where $B$ is the hypergraph incidence matrix and $f,g$ are entrywise nonlinearities.
%
It is particularly interesting that this type of nonlinear Laplacian operator appears in many settings. For example, in the graph case, if  $f=\mathrm{id}$ and $g(x) = |x|^{p-1}\mathrm{sign}(x)$, then $L$ boils down to the graph $p$-Laplacian operator \cite{buhler2009spectral,elmoataz2008nonlocal,saito2018hypergraph,upadhyaya2021self}. Exponential- and logarithmic- based choices of $f$ and $g$ give rise to nonlinear Laplacians used to model  chemical reactions \cite{rao2013graph,van2016network} as well as to model consensus dynamics and opinion formation in hypergraphs \cite{neuhauser2021consensus}. Trigonometric functions such as $g(x) = \sin(x)$ are used to model network oscillators  \cite{battiston2021physics,millan2020explosive,schaub2016graph}. Entrywise powers and generalized (power) means are used for node classification  \cite{arya2021adaptive,ibrahim2019nonlinear,prokopchik2021nonlinear,tudisco2021nonlinear}, network centrality and clustering coefficients \cite{arrigo2020framework}.

In particular, in \cite{tudisco2021nonlinear} this type of hypergraph mapping is used to define generalized eigenvector centrality scores for hypergraphs which include as special cases hypergraph centralities based on tensor eigenvectors \cite{benson2019three}. Thus,  the proposed hypergraph core-periphery score can be interpreted as a particular hypergraph centrality designed specifically for core-periphery problems and gives mathematical support to the intuition that centrality measures for hypergraphs may be an indication of core and periphery~\cite{amburg2021planted}. 


The paper is structured as follows.
In \S\ref{sec:mot} we motivate the core-periphery concept.
\S\ref{sec:model} introduces a random model for generating 
 hypergraphs with 
core-periphery structure and relates the model to a maximum likelihood 
problem.
This is used to justify the optimization problem that we propose and analyse in \S\ref{sec:optimization}.
In
\S\ref{sec:graph} we briefly discuss relevant work in the 
graph setting, and in particular we review how the 
clique expansion can be used as a means to 
approximate a hypergraph with a graph.
Computational experiments are presented in \S\ref{sec:exp} and
concluding remarks are given in \S\ref{sec:conc}.
Appendix~\ref{app:proofs} contains a proof of our underpinning theoretical result.

The main contributions of this work are
\begin{itemize}
\item the optimization formulation (\ref{eq:cp_opt}) for core-periphery detection at the higher-order hypergraph level, motivated by the generative random hypergraph model
(\ref{eq:mod}) (see Theorem~\ref{thm:opt}),
\item existence and uniqueness theory for this optimization problem,
and a practical, globally convergent iteration scheme (Theorem~\ref{thm:main}),
\item an interpretation of this approach as a nonlinear spectral method
(Corollary~\ref{cor:spectral}),
\item a core-periphery profile definition for hypergraphs, allowing methods to be compared on real data sets (\S\ref{subsec:profile}).
\end{itemize}

\section{Motivation}
\label{sec:mot}
We consider a hypergraph $H = (V,E)$ with $|V|=n$ nodes and $|E|=m$ hyperedges.
We assume that the nodes and hyperedges are labelled from $1$ to $n$ and 
from $1$ to $m$, respectively.
So, each hyperedge $e \in E$ has the form 
$(i_1,\ldots,i_r)$, 
where $1 \le i_1 < i_2 < \cdots < i_r \le n$, and we refer to $r$, the number of nodes present, as the size of the hyperedge.
We let $B \in \RR^{n \times m}$ denote the corresponding 
incidence matrix, so 
$B_{ij} = 1$ if node $i$ is present in hyperedge $j$, and 
$B_{ij} = 0$ otherwise.

We are interested in the case where interactions
can be described via a core-periphery mechanism, and we seek an 
algorithm that can uncover this structure when it is present in the data.
Loosely, core nodes are those that enable
interactions, whereas peripheral nodes  
may only take part in an interaction if at least one core
node is also present.
Suppose first that the nodes may be split into 
disjoint subsets: the core and the periphery. 
In the strictest interpretation, we could argue that 
the resulting core-periphery hypergraph 
will consist of precisely those hyperedges that contain 
at least one core node.
More generally, we may argue that 
a hyperedge is more likely to arise if it contains 
at least one core
node. 
Moreover, we may extend the idea further to 
argue that a hyperedge (of a given size)  
is more likely to arise  
if it involves more core nodes.
A second direction in which we can relax this definition is to 
argue that instead of a binary split, the nodes can be 
ranked in terms of their coreness.

Using these ideas, which are discussed further in \S\ref{sec:graph},
we build on
notions that have been proposed and tested in the standard network context to 
motivate the new model and algorithm.

\section{Core-periphery hypergraph random model}
\label{sec:model}
We now propose a model that generates hypergraphs with a planted core-periphery structure, generalizing the logistic core-periphery random model for graphs introduced in \cite{tudisco2019core}. Rather than a binary classification model where a node is first assigned to the core or periphery, the proposed random model generates hypergraphs where a smooth transition between the core and the periphery sets is allowed.  As discussed in \cite{tudisco2019core}, we find this situation more realistic, as real-world interactions with a perfectly clear cut
separation between core and periphery sets are rare. 

We assume each node is assigned an integer index $i\in \{1,\dots,n\}$ 
with the convention that the planted core structure starts `at the top', i.e., we want nodes with smaller indices to be in the core and those with large indices to be peripheral.
Let $\sigma(\cdot)$ denote the sigmoid function,  so that 
$\sigma(x) = 1/(1+e^{-x})$. We introduce and study here the generative model with a fixed number of nodes and 
a fixed upper limit on the maximum hyperedge size, 
where each hyperedge $e$ exists with independent probability
\begin{equation}
P(e\in E) =  \sigma \big( \xi(e)  \mu(e) \big). 
\label{eq:mod}
\end{equation}
Here $\xi(e)$ is a function that decreases with the size of the hyperedge, for example $\xi(e) = 1/|e|$, and $\mu(e)$ is a function that attains large values if the nodes in $e$ are near the core, i.e.\ they have small indices. For $q\geq 1$, an example choice of $\mu$ is
$$
\mu(e) = \mu_q(e) := \Big(\sum_{i\in e}{\Big(\frac {n-i}n\Big)^q}\Big)^{1/q}, 
$$
which is a smooth approximation of $\mu(e) = \max_{i\in e}\frac{n-i}{n}$, obtained for $q\to \infty$. Notice that $\mu_q$ corresponds to the $q$-norm of the vector with entries $(1-1/n, 1-2/n, \dots)$ restricted to the hyperedge $e$. We will often consider such $\mu_q$ in practice as it is directly connected to the optimization framework described in \S\ref{sec:optimization}. 

\begin{figure}[t]
    \centering
    \includegraphics[width=\textwidth]{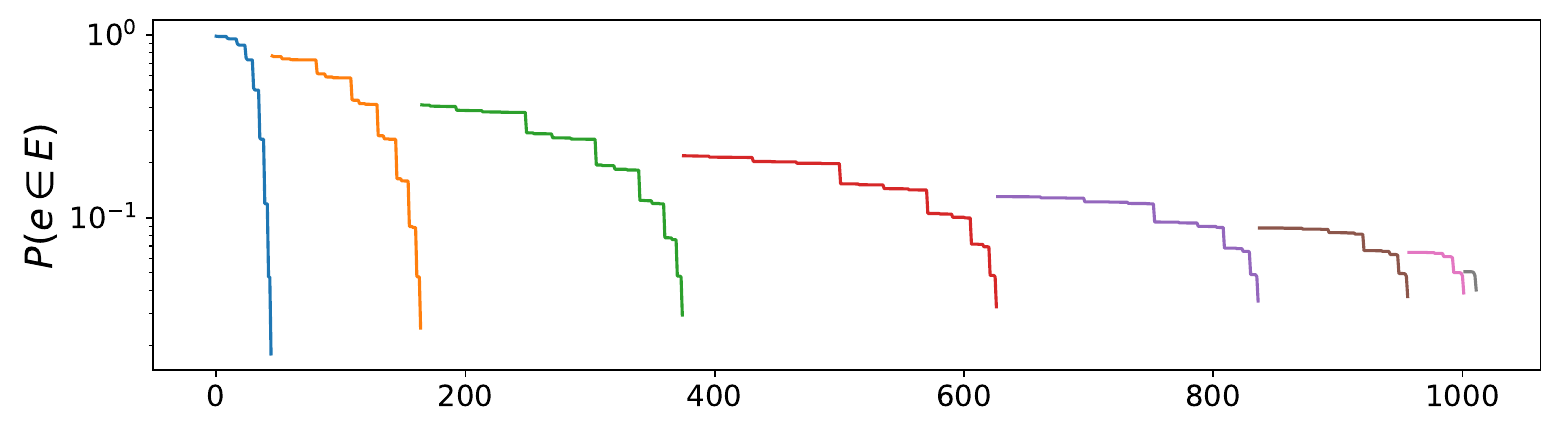}
    \caption{Hyperedge probability  distribution $P(e\in E)$ as defined in \eqref{eq:mod}, for a random hypergraph with $n=10$ nodes and  $2^n-(n+1) = 1013$ possible hyperedges. Here, $\xi(e) = 1/|e|$ and $q = 10$.}
    \label{fig:edge_probabilities}
\end{figure}
With this choice of $\mu$ and $\xi$, a hyperedge is more likely to exist when it contains a small number of nodes, at least one of which is part of the core. Figure \ref{fig:edge_probabilities} illustrates the behaviour of $P$ as a function of the edge $e$ on a random hypergraph with $n=10$ nodes and $2^n-(n+1) = 1013$ possible hyperedges given by all possible subsets $e$ of $\{1,\dots,n\}$ with $|e|\geq 2$.  In the figure, the hyperedges on the $x$-axis are sorted in lexicographical order and grouped by their size: 
\def\r#1{\rotatebox[origin=t]{67}{$#1$}}
$$
\small
\r{\{1,2\}}\r{\{1,3\}}\dots\r{\{1,n\}}\r{\{2,3\}}\r{\{2,4\}}\dots\r{\{2,n\}}\qquad \r{\{1,2,3\}} \r{\{1,2,4\}}\dots\r{\{1,2,n\}}\r{\{1,3,4\}}\r{\{1,3,5\}}\dots\r{\{1,3,n\}} \dots 
$$
We use different colors to distinguish between edge size.

We now turn to the inverse problem where it is required 
to discover a node ordering 
 that 
reveals core-periphery structure. 
We will let ${\cal P}$
denote the set of all permutations of $\{1,2,\ldots,n\}$. 
So $v \in {\cal P}$ is a vector $v \in \RR^n$ with distinct
integer elements between $1$ and $n$.
We will associate 
$v \in {\cal P}$ with a node reordering  such that 
node $i$ is mapped to node $v_i$.
Hence, if $v_j = 1$ and $v_k = n$, then,  according to this new ordering,
node $j$ is the most core and node $k$ 
is the  most peripheral.

In this framework, given a hypergraph, it is reasonable to choose an ordering that maximizes the 
likelihood of the hypergraph arising 
under the model 
(\ref{eq:mod}).
This viewpoint has been found useful for various graph models and structures \cite{GHZ21,RDRG}.
In the next theorem, we show that the resulting maximum likelihood 
problem can be converted into a 
discrete optimization problem that turns out to be
amenable to relaxation.

\begin{theorem}
\label{thm:opt}
For a given hypergraph,
a permutation vector $v \in {\cal P}$ corresponds to a maximum likelihood node reordering under the model 
(\ref{eq:mod}) 
if and only if
it maximizes the objective function
$
\sum_{e \in E} \xi(e) \mu(e)
$.
\end{theorem}

\begin{proof}
Under the model (\ref{eq:mod}), the likelihood is
\[
\prod_{e \in E} \sigma\big( \xi(e) \mu(e) \big) \, 
 \prod_{e \in E'} \left( 1 - \sigma\big(\xi(e) \mu(e)\big) \right),
 \]
 where 
 $E'$ denotes the complement of $E$; that is, the hyperedges that are not present.
We may rewrite this likelihood as 
\[
\prod_{e \in E}
\frac{
\sigma\big( \xi(e) \mu(e) \big)
}
{
1 - \sigma\big( \xi(e) \mu(e) \big)
}
\,  
 \prod_{\text{all~edges}} \left( 1- \sigma\big( \xi(e) \mu(e) \big)  \right).
 \]
 Now, the second product is independent of the node ordering. So we solve the problem by maximizing the first product, which may be written
 \[
\prod_{e \in E}
e^{\xi(e) \mu(e)}.
 \]
 Finally, taking a logarithm shows that the original problem is equivalent to maximizing 
 $
\sum_{e \in E} \xi(e) \mu(e)
$.
   \end{proof}

\section{Core-periphery detection via cost function optimization}\label{sec:optimization}
Motivated by Theorem~\ref{thm:opt}, 
we propose a model based on the optimization of a nonconvex core-periphery cost function.  Our goal is to determine a core-periphery nonnegative score vector $x^\star$ that assigns large value to nodes in the core and small values to those in the periphery. Clearly, such a vector is `scale invariant' in the sense that any positive rescaling of a core-periphery score vector corresponds to the same core-periphery assignment. For this reason,  given a hypergraph $H = (V,E)$, we formulate the core-periphery detection problem as the following norm-constrained optimization problem:
\begin{equation}\label{eq:cp_opt}
    \max_x f(x) \quad  \text{s.t. } \|x\|_p = 1 \text{ and } x\succeq 0,
    \quad \text{ where }
    \quad 
f(x) = \sum_{e \in E} \, \xi(e) \, \|x|_e\|_{q} \, ,
\end{equation}
$x \succeq 0$ denotes a vector with nonnegative entries, and $x|_e$ denotes the restriction of $x$ to the nodes in the hyperedge $e$, i.e.\ $x|_e =(x_{i_1}, \dots, x_{i_r})$ if $e = (i_1,\dots,i_r)$. As before,  $\xi(e)$ is a function that assigns a weight to the hyperedges in $H$. In practice, $\xi$ may be both `data-driven',  in the sense that it may 
incorporate a weight $w(e)$ of the hyperedge in the input dataset, and `model-oriented', in the sense that it should decrease with the size of the hyperedges (e.g., $\xi(e)$ proportional to $1/|e|$) so to take into account our modeling assumption that, in the hypergraph core-periphery model, hyperedges with more nodes should make a smaller contribution to the assignment of the core-score. 

Note that, as $x\succeq 0$, when $q$ is large we have $\|x|_e\|_{q}\approx \max_{i \in e} x_i$ and thus $f(x)$ is large if many hyperedges contain at least one node with large core-periphery score value. This may be interpreted as the core-periphery analog of the widely used `all-or-nothing' definition of hypergraph cut function in the context of hypergraph clustering and, in general, of cut-based hypergraph problems \cite{hein2013total,veldt2020hypergraph}. However, similarly to the hypergraph cut setting, when considering a core-periphery score one may want to account for the fact that hyperedges may contain more than  one core node, and give these hyperedges a greater importance.
We note that this is somewhat automatically obtained by our choice of smooth function $f$. In fact, while in the graph setting  the non-smooth limit case $q\to \infty$ is to be preferred \cite{tudisco2021nonlinear} as each edge contains exactly either one, two or no core nodes, 
we argue that large but finite values of $q$ are better suited to  hypergraphs. In fact, when $1\ll q<\infty$, the cost function $f(x)$ naturally handles possible ambiguity due to the presence of hyperedges with more than two nodes in the core: While the infinity norm $\|x|_e\|_\infty$ is large if there is at least one core node in $e$ but ignores the presence of  a larger number of core nodes, $\|x|_e\|_q$ (for large but finite $q$), is large when there is at least one core node in $e$  but grows when the hyperedge contains a larger  number of such nodes. 


Although \eqref{eq:cp_opt} is a nonconvex optimization problem, we show below that if we are interested in entry-wise positive solutions, then it admits a unique solution which we can always compute to an arbitrary accuracy via a linearly convergent method, provided that $p>q$. We move the relatively long proof of this result to Appendinx \ref{app:proofs}.

\begin{theorem}\label{thm:main}
If $p>q>1$ then \eqref{eq:cp_opt} has a unique entry-wise positive solution $x^\star$. Moreover, for any entrywise positive starting vector, 
the iterative scheme
\begin{itemize}
    \item $y \gets \Diag(x)^{q-1}  B \Xi(B^\top x^q)^{\frac 1 q-1}$
    \item $x \gets (y/\|y\|_{p^*})^{\frac 1 {p-1}}$
\end{itemize}
where $p^* = p/(p-1)$ is the H\"older conjugate of $p$ and where $\Xi$ is the  diagonal $|E|\times |E|$ matrix with diagonal values $\xi(e)$, 
converges to $x^\star$  with the linear rate of convergence $O(|q-1|/|p-1|)$. 
\end{theorem}

It is interesting to note that the vector computed in this way can be interpreted as 
a `nonlinear eigenvector  centrality' for the nodes in the hypergraph, as per the model introduced in \cite{arrigo2020framework,tudisco2021node}. In fact, the following direct corollary of the previous result shows that the core score $x^\star\succ 0$ solution to \eqref{eq:cp_opt} coincides with an eigenvector of a nonlinear hypergraph Laplacian operator of the form $L(x) = B\Xi g(B^\top f(x))$, for particular nonlinear choices of $f$ and $g$. 

\begin{corollary} \label{cor:spectral}
Let $x^\star\succ 0$ be the unique solution of \eqref{eq:cp_opt}. Then, $x^\star$ is the unique nonnegative eigenvector of the nonlinear eigenvector problem 
\begin{equation}\label{eq:eigen_form}
    B\Xi g(B^\top f(x))  = \lambda x
\end{equation}
with $g(x)=x^{-1 + 1/q}$ and  $f(x) = x^{q/(p-q)}$. 
\end{corollary}
\begin{proof}
Let $F(x) =  \Diag(x)^{q-1}  B \Xi(B^\top x^q)^{\frac 1 q-1}$. Assume $x\succ 0$. From the fixed point identity $x= (F(x)/\|F(x)\|_{p^*})^{1/(p-1)}$ we have $F(x)^{1/(p-1)} = \|F(x)\|_{p^*}^{1/(p-1)}x$. Multiplying this identity entrywise by $x^{(1-q)/(p-1)}$ on the left and then taking the $(p-1)$-th of power both sides, we get
$$
B\Xi(B^\top x^q)^{\frac 1 q -1} = \|F(x)\|_{p^*}\,x^{p-q}\, .
$$
Finally, the change of variable $x \mapsto x^{1/(p-q)}$  shows that $x^\star \succ 0$ is the limit of the iterative scheme in Theorem \ref{thm:main} if and only if $x^\star$ is such that $B\Xi g(B^\top f(x^\star))  = \lambda x^\star$ with $g(x)=x^{1/q-1}$,  $f(x) = x^{q/(p-q)}$ and $\lambda = \|F((x^\star)^{1/(p-q)})\|_{p^*}>0$. 
\end{proof}

Before moving on, we briefly point out how Theorem \ref{thm:main} compares with the main theorem in \cite{tudisco2021node} in view of the corollary above. A direct consequence of  Corollary \ref{cor:spectral}  combined with \cite[Thm.~2.3]{tudisco2021node} shows  that the nonlinear power method proposed in \cite{tudisco2021node} for general nonlinear singular value problems, can be used to compute a solution to \eqref{eq:eigen_form}, provided $|p-q|\geq |q-1|$ and the bipartite graph representation of the underlying hypergraph is connected. Note that the condition on $p$ and $q$ in this case boils down to $p\geq 2q-1$, when $p,q>1$. Using a different argument, Theorem \ref{thm:main} shows that for the particular choice of $g$ and $f$ which correspond to the core-periphery optimization problem \eqref{eq:cp_opt}, the less stringent condition $p>q$ is enough to ensure convergence of the proposed fixed point iteration to the solution of \eqref{eq:eigen_form}, without any requirement on the topology of the hypergraph. 

\section{Comparison with the graph setting}
\label{sec:graph}

Following the seminal work by Borgatti and Everett \cite{borgatti2000models}, over the years several models for core-periphery detection on graphs have been developed, including methods based on degree and eigenvector centralities \cite{mondragon2016network,rombach2017core}, rank-1 approximations \cite{minres} and the optimization of a core quality fictional  \cite{rombach2017core}. As observed in \cite{tudisco2019core}, several of these methods can be cast as the optimization of a core-periphery kernel function similar to \eqref{eq:cp_opt} and a competitive core-periphery detection method available for graphs is obtained there by means of a nonlinear spectral method, which corresponds to the graph version of the method we propose here.  All these methods can be directly applied to hypergraphs after  a `flattening' or `projection' step, where the whole higher-order graph is
approximated by a standard graph. A widely-used projection approach is the so-called (linear) `clique-expansion'   \cite{agarwal2006higher,carletti2020random,HdK21,
rodri2002laplacian,rodriguez2003laplacian,rodriguez2009laplacian,zhou2007hypergraph}, where hyperedges in $H$ are replaced by cliques in the flattened graph $G_H$, whose adjacency matrix $A_H$ therefore becomes 
\begin{equation}\label{eq:clique-expansion-adjacency}
    (A_H)_{ij} = \sum_{e: \, i,j\in e}w(e) 
\end{equation}
 with $w(e)$ the weights of the original hypergraph. While this is perhaps the most popular projection method, other approaches are possible, including clique averaging \cite{agarwal2005beyond}, where  the weights $w(e)$ in the sum \eqref{eq:clique-expansion-adjacency}  are averaged with generalized mean functions,  connectivity graph expansion \cite{banerjee2021spectrum,de2021phase}, where the weights in the clique expansion are based on hyperedge degrees, for example replacing $w(e)$ with $1/(|e|-1)$ in \eqref{eq:clique-expansion-adjacency}, and the star expansion \cite{zien1999multilevel}, where the flattened graph is obtained by introducing new vertices for each hyperedge, which are then connected according to the hypergraph structure.

However, graph core-periphery detection methods applied to a projected hypergraph may lead to poor core-periphery assignments on the original hypergraph. For example,  when applied to the clique-expanded graph, the nonlinear spectral method proposed in \cite{tudisco2019core} computes the global optimizer of 
$$
\sum_{ij}(A_H)_{ij}(|x_i|^q + |x_j|^q)^{1/q} = \sum_{e\in E}\Big\{ \sum_{i,j \in e} w(e) (|x_i|^q + |x_j|^q)^{1/q}\Big\}\, .
$$
While this objective function reduces to the proposed hypergraph optimization problem \eqref{eq:cp_opt} when all hyperedges contain exactly two nodes, the two optimization problems are significantly different in the general case. 
In particular, the hypergraph flattening loses track of the original hyperedges and thus only measures core and periphery structure in a pairwise fashion. Hence, approaches based on hypergraph flattening may fail to assign a correct score to the nodes even in simple hypergraph examples, as shown in the next section. 

\section{Experiments}
\label{sec:exp}
In this section we perform experiments on a range of hypergraphs in order to validate the performance of the proposed nonlinear eigenvector method for the optimization of \eqref{eq:cp_opt}, which we will denote as \HyperNSM. In all our experiments we use the scaling function $\xi(e) = 1/|e|$ and we choose $q=10$ and $p=11$ in \eqref{eq:cp_opt}. For a hyperedge weight $w(e)>0$, we set $\xi(e) = w(e)/|e|$. 

We compare our method with a number of existing approaches. One is the Union of Minimal Hitting Sets (\UMHS{}) method by Amburg et al.\ \cite{amburg2021planted}, which is a greedy method designed to recover planted hitting sets in hypergraphs. The method is inherently local and thus requires several random initializations out of which the best assignment is selected. In our experiments we perform five random restarts due to time limitations. 
Two additional baselines are core-periphery detection algorithms for graphs, applied to the clique-expanded graph in (\ref{eq:clique-expansion-adjacency}): 
\begin{enumerate}
    \item \BorgattiEverett: The method by Borgatti and Everett \cite{borgatti2000models}, which is the pioneering approach for core periphery detection in graphs
    \item \GraphNSM: The nonlinear spectral method by Tudisco and Higham \cite{tudisco2019core}, which was shown to be a highly competitive core-periphery detection method for graphs.
\end{enumerate}
 
 \subsection{Hypergraph core-periphery profile}\label{subsec:profile}
In order to evaluate the quality of the core-periphery assignment we introduce a generalization of the core-periphery profile for graphs \cite{della2013profiling,tudisco2019fast} which we define as follows.
For any subset of nodes $S\subseteq V$ consider the quantity
$$
 \gamma(S) = \frac{\text{\# edges all contained in }S}{\text{\# edges with at least one node in }S}, 
$$
or its weighted version:
$$
\gamma(S)  =  \Big(\sum_{e: \,  e\subseteq S}\xi(e)\Big) \Big(\sum_{e: S\cap e \neq \emptyset}\xi(e)\Big)^{-1}.
$$
Given a core-periphery score vector $x\succeq 0$ with distinct entries  
the \emph{hypergraph core-periphery profile} is the function $\gamma(k)$ that to any $k\in \{1,\dots, n\}$ associates the value $\gamma(S_k(x))$ where $S_k(x)$ is the set of $k$ nodes with smallest core-periphery score in $x$. Given its definition, $\gamma(S)$ is small if $S$ is largely contained in the periphery of the hypergraph.  Thus a hypergraph has a strong core–periphery
structure revealed by a core–score vector $x$ if the corresponding profile $\gamma(k)$ attains
small values as $k$ increases from one, and then grows drastically as $k$ crosses some
threshold value $k_0$, which indicates that the nodes in $V\setminus S_{k_0}(x)$ form the core.

\subsection{Hyperplane and hypercycle results}\label{subs:examples}
Consider the hypergraph in the left and center panels of Figure~\ref{fig:hyperplane}. 
Because of its airplane outline, we will refer to this hypergraph
as a  `hyperplane'. 
Intuitively the clique-expanded approach will not work here as the large `wing' hyperedge gives rise to a fully connected subgraph in $G_H$ which would correspond to a core set in the projected graph, which is not present in the original hypergraph.
This is confirmed by the central plot in Figure \ref{fig:hyperplane}, where the nodes in the hypergraph are shaded according to the core score computed with \GraphNSM{} on the clique-expanded graph. 
By contrast, 
\HyperNSM\ gives a more satisfactory result where 
each hyperedge has exactly one of the top-two core nodes, and the non-overlapping wing nodes are regarded as peripheral
(left panel in Figure \ref{fig:hyperplane}).
 This difference is also reflected in the rightmost panel of 
 Figure~\ref{fig:hyperplane} where we compute the hypergraph core-periphery profiles introduced in \S\ref{subsec:profile}. The score assigned by \HyperNSM{} identifies a two-node core, while no hypergraph core-periphery structure seems to be captured by the 
 graph method applied to the clique-expansion.

\begin{figure}[t!]
    \centering
    \begin{subfigure}[t]{1\textwidth}
    \centering
    \includegraphics[width=\textwidth]{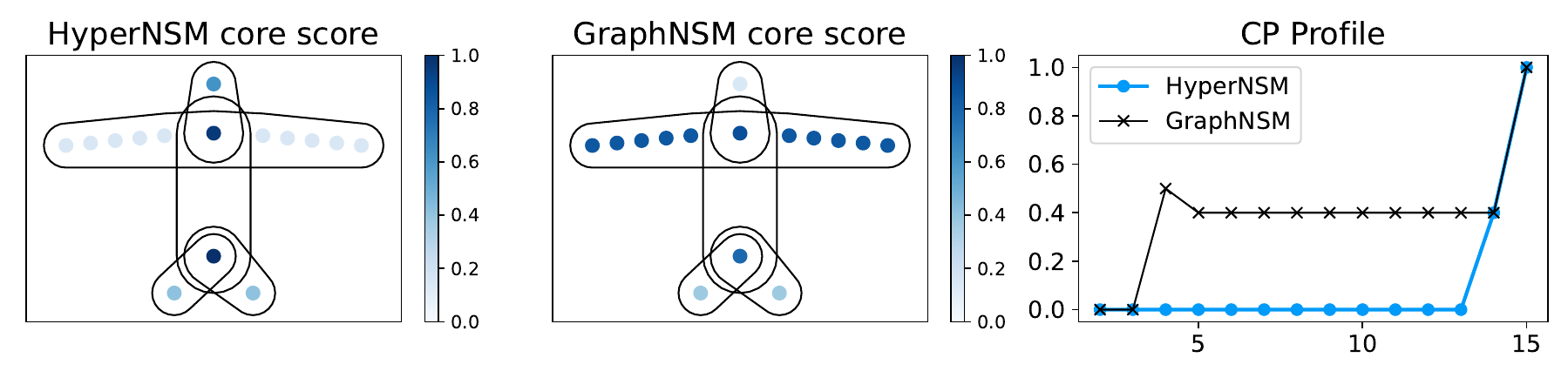}
    \caption{Hyperplane with five hyperedges and 15 nodes.}
    \label{fig:hyperplane}
    \end{subfigure}
    
    \vspace{1em}
    
    \begin{subfigure}[t]{1\textwidth}
    \centering
    \includegraphics[width=\textwidth]{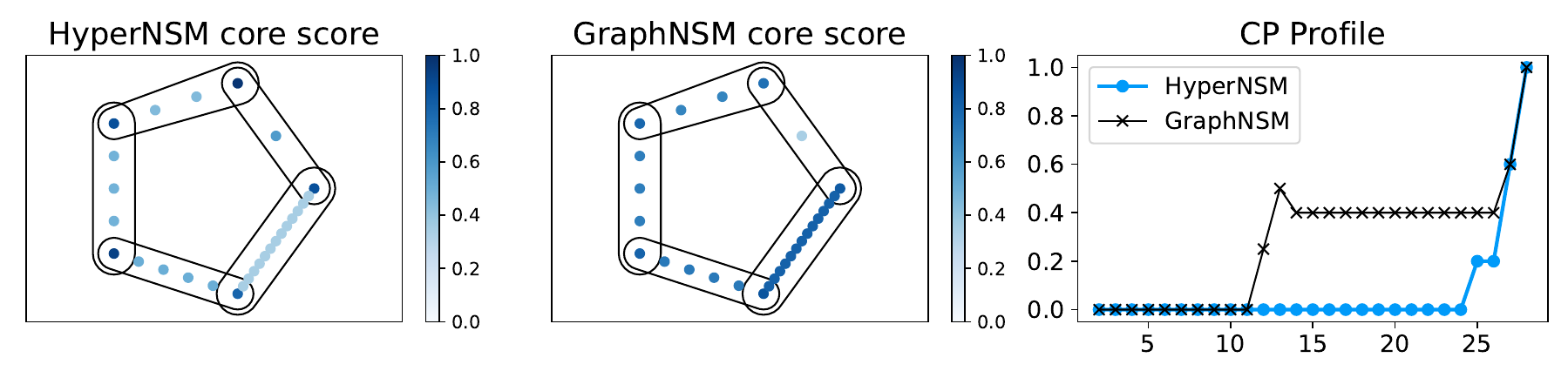}
    \caption{Hypercycle with five hyperedges and 28 nodes.}
    \label{fig:hypercycle}
    \end{subfigure}
   
    \caption{Left and central panels: hypergraph drawing with nodes colored according to the core score obtained by means of the proposed \HyperNSM{} approach (left) and the purely graph method \GraphNSM{} (center). Right panel: hypergraph core-periphery profile corresponding to the two core score assignments, i.e.\ $\gamma_x(k)$ as defined in \S\ref{subsec:profile}, plotted as a function of $k$, for the two core score vectors $x$.}
\end{figure}

We observe similar behaviour on the `hypercycle' hypergraph
shown in Figure~\ref{fig:hypercycle}.
Here, each hyperedge shares exactly two nodes with exactly two other hyperedges, in a periodic fashion,
and we have five hyperedges of size 3, 4, 5, 6 and 15. 
 As one hyperedge is much larger than the others, the clique expansion approach assigns all of its nodes to the core. 
 HyperNSM{}\ instead assigns high coreness to each of the 5 `overlap' nodes. 
In the right panel we see that the hypergraph core-periphery profile starts to increase when the first two
 overlap nodes are included, since every hyperedge contains exactly two of these nodes.




\subsection{Real-world datasets with planted core} \label{sec:experiments_planted_core}
We consider here two real-world hypergraph datasets---\textit{W3C} and \textit{Enron}---with a planted core set that arises directly 
from the data collection process, as discussed in \cite{amburg2021planted}. Both the datasets are  email hypergraphs, in which nodes are email addresses. Each hyperedge records a set of email addresses that appear on the same email.  
Table~\ref{tab:email_datasets_statistics} reports a summary of statistics of the two datasets.



\begin{table}[t]
    \centering
\begin{tabular}{lccccccccl}
\toprule
 & \#nodes & \#edges & \#nodes in core & \multicolumn{3}{c}{$w(e)$} & \multicolumn{3}{c}{$|e|$}\\
\cmidrule(lr){5-7}\cmidrule(lr){8-10}
 &          &        &   & max  & mean & std    & max & mean & std \\
\midrule
\textit{Enron} $\quad$ & 12722 & 5734 & 132 & 419 & 2.73 & 9.07 & 25 & 5.25 & 5.1 \\
\textit{W3C} & 15458 & 19821 & 1509 & 1 & 1 & 0 & 25 & 2.22 & 0.98 \\
\bottomrule
\end{tabular}
\caption{Basic statistics for the two email datasets with planted hitting set}\label{tab:email_datasets_statistics}
\end{table}

These datasets are characterized by a planted core, which consists of a `hitting set' of nodes such that every hyperedge contains 
at least one of these nodes. (In practice, the data was collected by 
examining the email accounts of the hitting set.)
 The greedy algorithm \UMHS{}, based on the union of minimal hitting sets, was designed in \cite{amburg2021planted} with the aim of recovering a planted core of this type. Even though this definition of core differs from the one we consider in this paper, the experiments illustrated in Figures \ref{fig:core_profile_email} and \ref{fig:core_intersection_email} show that the performance of  \HyperNSM{} is on par with that of \UMHS{} (with five random restarts). 

\begin{figure}[t!]
    \centering
    \includegraphics[width=\textwidth]{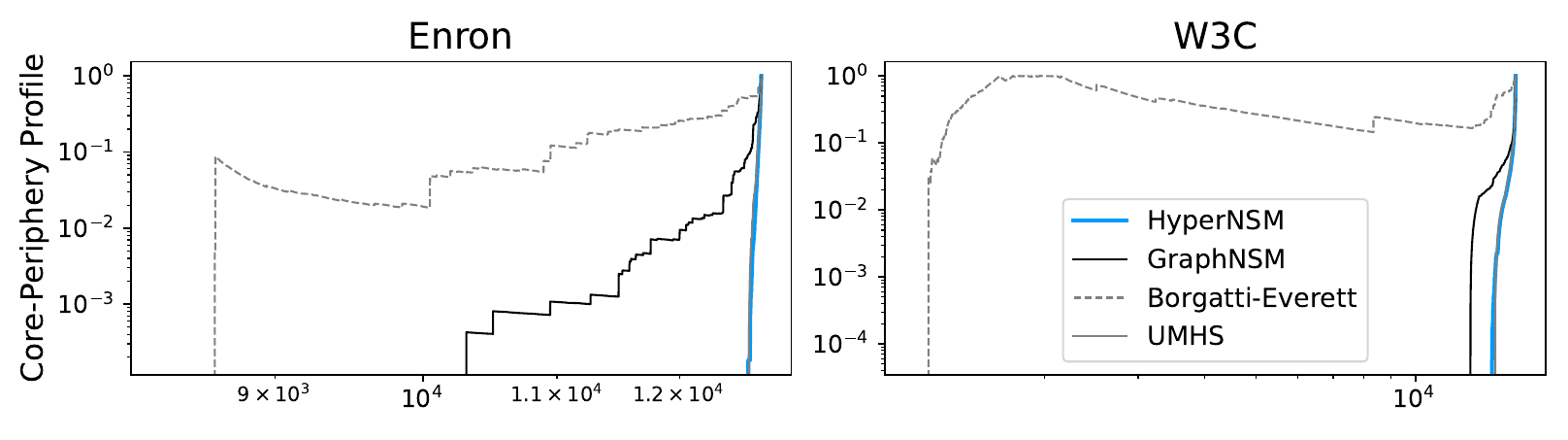}
    \caption{Core-periphery profile for different core-periphery detection methods. The core-periphery profiles for \HyperNSM{}\ and \UMHS{}\ are visually indistinguishable in both cases.}
    \label{fig:core_profile_email}
\end{figure}
\begin{figure}[t!]
    \centering
    \includegraphics[width=\textwidth]{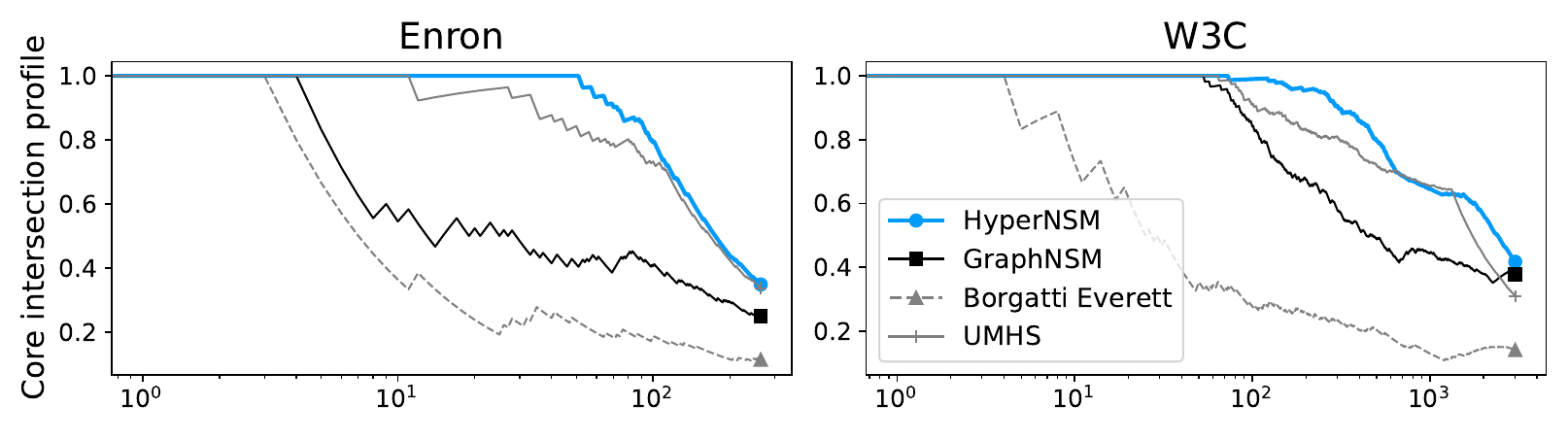}
    \caption{Intersection between the top $k$ nodes as ranked by different core-periphery methods and the planted core set, for increasing values of $k$ ($x$-axis).}
    \label{fig:core_intersection_email}
\end{figure}

Precisely, Figure \ref{fig:core_profile_email}  shows the core-periphery profiles corresponding to the core score computed with the four methods \HyperNSM{}, \GraphNSM{}, \BorgattiEverett{} and \UMHS{}. As \UMHS{} is originally designed to output a ordered list of nodes such that the nodes at the top of the list belong to the core, we treat this list as a ranking and use the corresponding ordering as a discrete core-score vector for \UMHS{}. 
The core-periphery profiles for \HyperNSM{}\ and \UMHS{}\ are visually indistinguishable in both cases.

Figure \ref{fig:core_intersection_email}, instead, shows the `core intersection profile', which is computed as follows. 
Given the planted core set $C$, for any set $S$ let 
$$
\iota(S) = |S\cap C|/|S|\, .
$$
Similarly to the core-periphery profile, the core intersection profile corresponding to a vector $x\succeq 0$ is the function $\iota_x(k)$ which to any $k$ associates the value $\iota(S_k(x))$ where $S_k(x)$ is the set of top $k$ nodes according to $x$. Thus, the core-score vector $x$ well-captures the planted core $C$ if its corresponding core intersection profile remains $\approx 1$ as $k$ is increased up to the size of the planted core. 

While having comparable performance 
to \HyperNSM{} 
on these test datasets, \UMHS{} is significantly more expensive than the other algorithms, as shown by the computational times in Table~\ref{tab:execution_time}. For \UMHS{} we report the execution time for one run as implemented in \cite{amburg2021planted}, where the method finds a hitting set processing the hyperedges in a random order, then prunes the set to reduce it to a minimal hitting set. For the other methods, we show the time elapsed for the relative norm of the difference of two consecutive iterates to reach the tolerance $\texttt{1e-8}$. In terms of computational complexity, the cost of  each iteration of \HyperNSM{} is dominated by the matrix vector products $Bv$ and $B^\top v$, which are linear in the number of input data, i.e., the cost of each step is $O(\sum_e |e|)$. Thus, for sparse hypergraphs, the method is fast.  For \GraphNSM{}, the cost of each iteration is dominated by the matrix vector product $A_Hv$. As the number of nonzero entries in $A_H$ is $O(\sum_e |e|)$, \GraphNSM{} has comparable computational complexity to \HyperNSM{}. However, \GraphNSM{} also requires $A_H$ to be formed from $H$, which may be expensive and memory demanding when very large hyperedges are present.  The cost per iteration of \BorgattiEverett{} is also linear in the number of nonzeros of $A_H$ but, as this is a purely-linear iteration (no entriwise powers are required) it is in practice faster than the nonlinear spectral iterative counterparts. 

\begin{table}[t]
    \centering
    \begin{tabular}{lccccc}
    \toprule
    & & \HyperNSM{} & \GraphNSM{} & \BorgattiEverett{} & \UMHS{}   \\
    \midrule
    \textit{Enron} & & 1.61 & 4.99 & 0.33 & 29.63\\
    \textit{W3C} & & 3.32 & 3.74 & 0.04 & 1269.17\\
    \bottomrule
    \end{tabular}
    \caption{Execution time (sec) for different methods on the two email datasets from \S\ref{sec:experiments_planted_core}. The table shows time for one run of \UMHS{}. The other methods are run until the relative norm of the difference of two consecutive iterates is smaller than $\texttt{1e-8}$. }
    \label{tab:execution_time}
\end{table}

We also note that \UMHS{}\ is based on the strict definition that
the core nodes form a 
hitting set; here every hyperedge must contain at least one core node.
For this reason, the algorithm is not well-suited to more general 
data sets where there is not a perfect planted core.
\HyperNSM{} is designed to  tolerate spurious or noisy 
information---note that the model (\ref{eq:mod}) may admit hyperedges
involving only peripheral nodes, albeit with low probability.

As a final experiment, we show in Figure \ref{fig:spyplots_email} colored sparsity plots for the clique-expanded graph of the two email hypergraph datasets. Each sparsity plot shows the nonzero entries of the weighted adjacency matrix  of the clique-expanded graph, as defined in \eqref{eq:clique-expansion-adjacency}.  Nonzero entries are shaded according to their relative value; darker $(A_H)_{ij}$ correspond to larger edge weights. Each column in the figure shows the colored sparsity pattern obtained by permuting the entries of the matrix according to the core-periphery score vector obtained with one of the methods considered. This figure further highlights how the core-periphery detection problem in hypergraphs fundamentally differs from the same problem on the projected graph. Although 
\GraphNSM{} fails to recover the planted hypergraph cores, as we saw in 
Figure~\ref{fig:core_intersection_email},
it finds more compact core structures than \HyperNSM{} on the
clique expansion graph (\ref{eq:clique-expansion-adjacency}). 

\begin{figure}[t!]
    \centering
    \includegraphics[width=\textwidth,clip,trim=2.6cm 8.8cm 2.4cm 7.8cm]{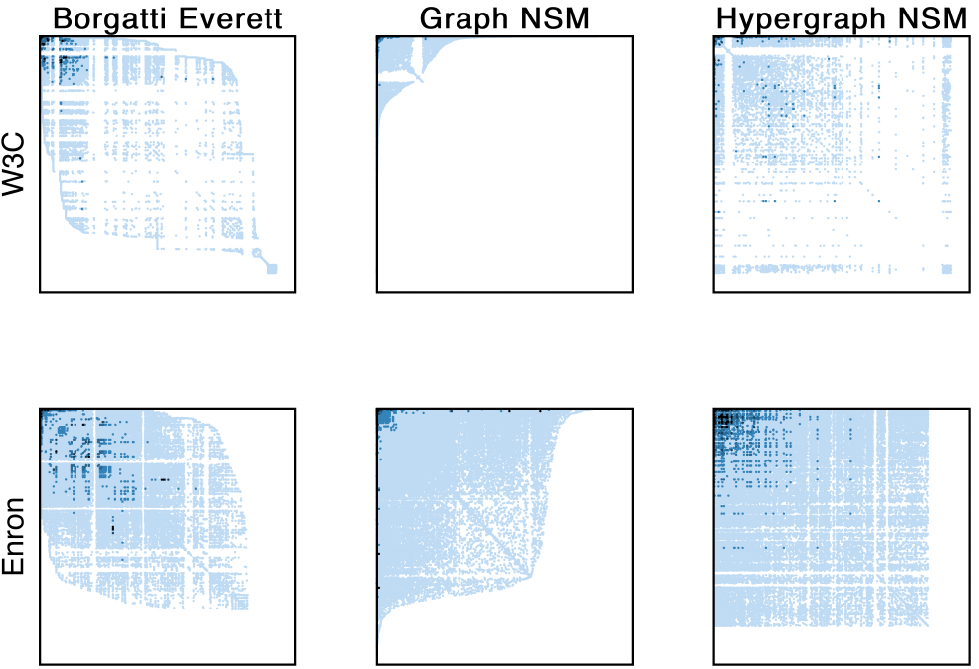}
    \caption{Colored sparsity plots of the clique-expanded graph's adjacency matrix. Nonzero elements in the matrix are colored according to their value (the larger the darker) and are permuted according to the entries of the core-periphery score vector of three different methods.}
    \label{fig:spyplots_email}
\end{figure}


\subsection{Real-world datasets with no available core information}

Next, we test
\HyperNSM{},
\UMHS{}\ and
\GraphNSM\ 
on a set of hypergraphs where the presence of a core set is not known 
a priori. Basic details about the datasets used in this section are summarized in Table~\ref{tab:datasets_statistics}.

\begin{table}[t]
\centering
\begin{tabular}{lcccccccc}
\toprule
 & \# nodes & \# edges & \multicolumn{3}{c}{$w(e)$} & \multicolumn{3}{c}{$|e|$}\\
\cmidrule(lr){4-6}\cmidrule(lr){7-9}
 &          &          & max  & mean & std    & max & mean & std \\
\midrule
\textit{Cora} & 2708 & 1579 & 1 & 1 & 0 & 6 & 4.03 & 1.02 \\
\textit{Citeseer} & 3306 & 1079 & 1 & 1 & 0 & 27 & 4.2 & 2.02 \\
\textit{Pubmed} & 19717 & 7963 & 1 & 1 & 0 & 172 & 5.35 & 5.67 \\
\inrule
\textit{NDC classes} & 1161 & 1090 & 2083 & 45.62 & 150.11 & 39 & 5.97 & 4.99 \\
\textit{NDC substances} & 5556 & 10273 & 2419 & 10.99 & 59.13 & 187 & 6.62 & 9.3 \\
\textit{Tags ask-ubuntu} & 3029 & 147222 & 1373 & 1.84 & 11.76 & 5 & 3.39 & 1.02 \\
\textit{Email EU} & 1005 & 25148 & 4875 & 9.36 & 50.35 & 40 & 3.56 & 3.4 \\
\bottomrule
\end{tabular}
\caption{Basic statistics for the real-world hypergraph datasets with no planted core. }\label{tab:datasets_statistics}
\end{table}

The first collection of datasets are co-citation hypergraphs \textit{Cora}, \textit{Citeseer} and \textit{Pubmed}  \cite{giles1998citeseer,mccallum2000automating,sen2008collective}. All nodes in the datasets are documents and hyperedges are based on co-citation (all papers cited in one manuscript form a hyperedge). These hypergraphs are unweighted. The second 
collection of hypergraphs is built starting from a  timestamped sequence of simplices,  as in  \cite{Benson-2018-simplicial}. Given a temporal sequence of simplices where each simplex is a set of nodes, we represent the dataset as a hypergraph with one hyperedge for each simplex (ignoring the time stamp), weighted with an integer counting how many times that hyperedge appears in the data. Nodes in \textit{Email EU} are email addresses at a European research institution and each hyperedge is a set of emails sent to one or multiple recipients. Nodes in the \textit{NDC substances} dataset are drugs, and hyperedges are formed by all the drugs corresponding to a  National Drug Code by the U.S.A.\ Food and Drug Administration, while the \textit{NDC classes} dataset is made out of one hyperedge  per drug and the nodes are class labels applied to the drugs. Finally, nodes in \textit{Tags ask-ubuntu} are tags and hyperedges are the sets of tags applied to questions on \texttt{https://askubuntu.com/}.

\begin{figure}[t]
    \centering
        \includegraphics[width=\textwidth,clip,trim=0 0 1.25em 0]{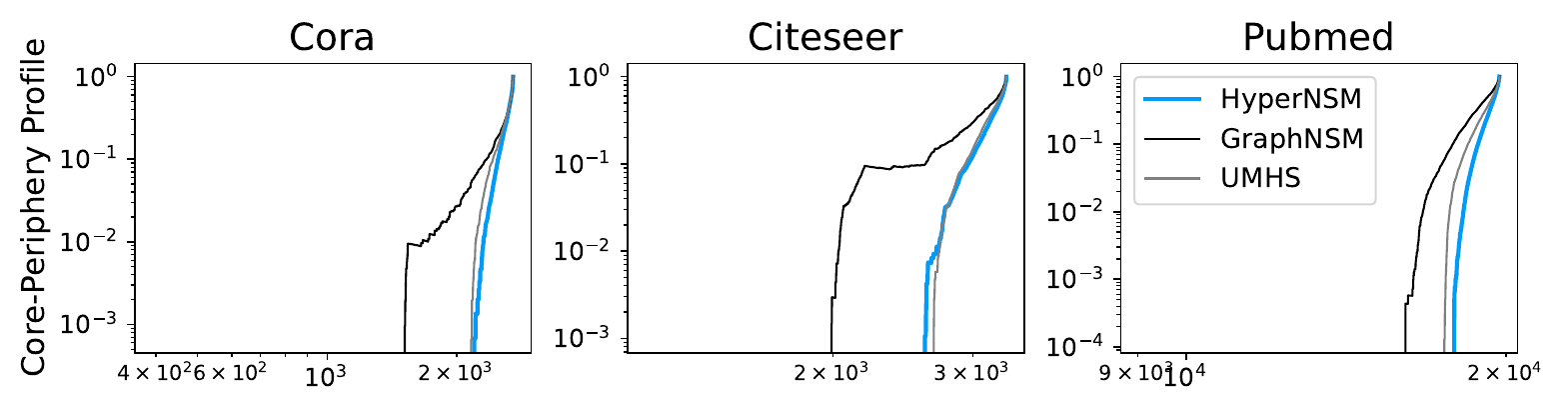}
    \caption{Core-periphery profiles of citation hypergraphs
    for three different core-periphery detection methods.}
    \label{fig:cp_profiles_citation}
\end{figure}

From the core-periphery profiles in Figure~\ref{fig:cp_profiles_citation} we see that for this citation data
both 
\HyperNSM{}\ and 
\UMHS{}\ find a convincing core structure.
The clique expansion based method, 
\GraphNSM, 
also gives a sharp core-periphery transition, albeit with a much 
larger core.
For the timestamped simplex data, however, 
as shown in 
Figure~\ref{fig:cp_profile_real_data} 
both the clique expansion and hitting set approaches fail to 
reveal core-periphery structure, whereas
\HyperNSM{}\ gives a sharp transition in each case.


\begin{figure}[t]
    \centering
    \includegraphics[width=\textwidth]{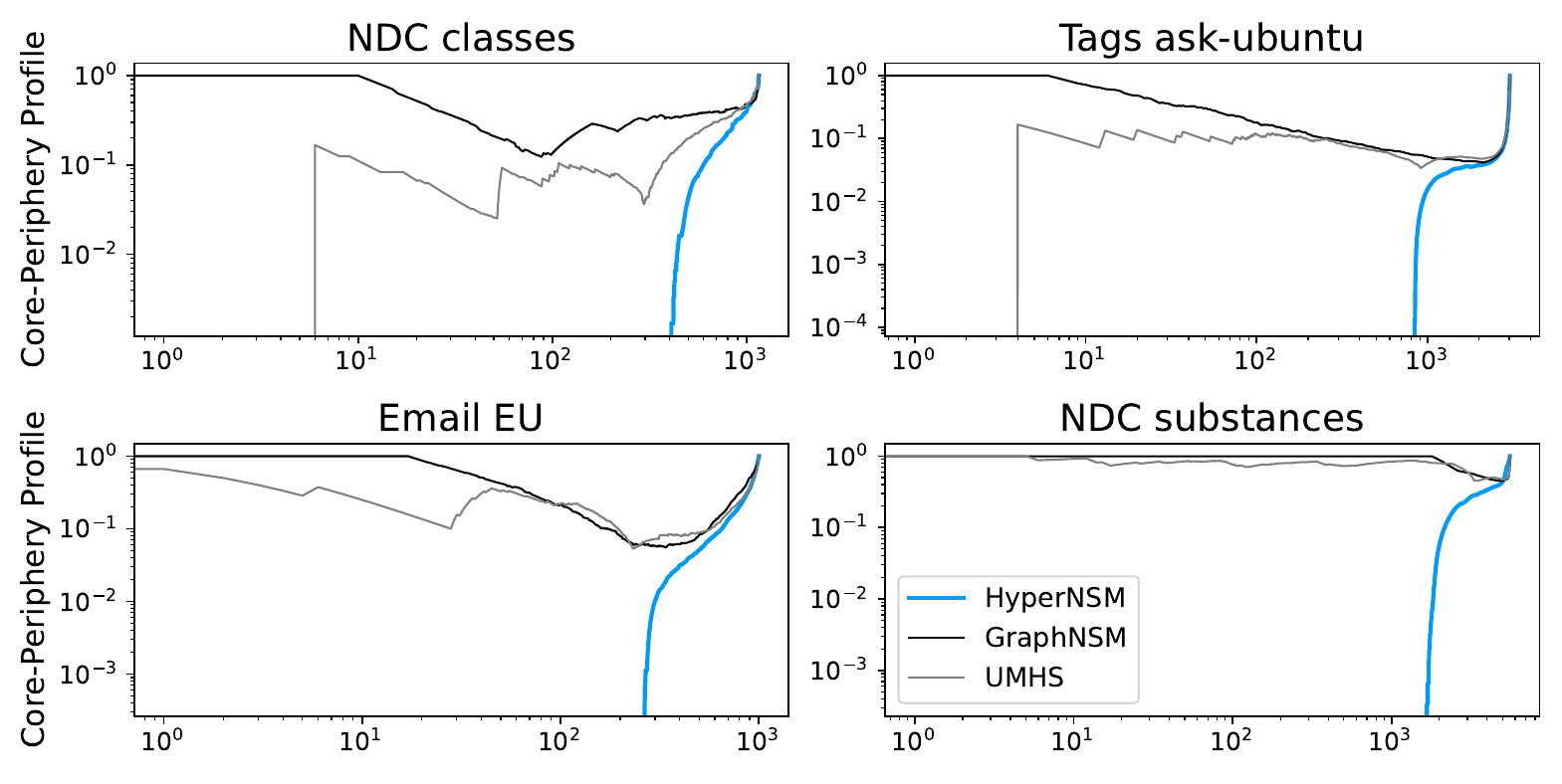}
    \caption{Core-periphery profile of four hypergraphs obtained from timestamped simplex data, for three different core-periphery detection methods.}
    \label{fig:cp_profile_real_data}
\end{figure}

\inew{}

\section{Conclusion}\label{sec:conc}
Our aim in this work was to extend existing graph-based 
core-periphery concepts and algorithms 
to the hypergraph setting, in order to account for  
the group-level interactions captured by many modern data
collection processes.
Our key take-home messages are that
\begin{itemize}[leftmargin=1em]
    \item the widely used approach of solving a suitable optimization problem can be generalized to the hypergraph case, and 
    a globally convergent iteration scheme is available,
    \item in this core-periphery setting, the general-purpose 
    recipe  
    ``flatten the hypergraph into a weighted clique expansion graph and apply a graph algorithm'' does not reduce the 
    computational complexity and does not lead
    to satisfactory results, 
    \item an extension of the graph core-periphery profile
    can be used to judge performance, and the new nonlinear spectral 
    method revealed core-periphery structure in a range of real datasets.
\end{itemize}

\bigskip
\section*{Data Statement}
The code used for the computations described here can be found at 
\href{https://github.com/ftudisco/core-periphery-hypergraphs}{https://github.com/ftudisco/core-periphery-hypergraphs} 
and the real datasets were obtained via the references cited in the text.

\appendix

\def\tilde{\widetilde}
\section{Proof of the main result}\label{app:proofs}

We devote this section to the proof of Theorem \ref{thm:main}. 
%
Let $f$ be defined as in \eqref{eq:cp_opt},  
with $\xi(e)$ a positive scaling function of the hyeperedges. 
As $f$ is positively one-homogeneous, i.e.,\ we have $f(\lambda x) = \lambda f(x)$ for all $\lambda>0$, the constrained optimization problem in \eqref{eq:cp_opt} coincides with the unconstrained optimization of $g$ defined as $g(x) = f(x/\|x\|_p) = f(x)/\|x\|_p$. Thus, $x^\star$ is a solution to \eqref{eq:cp_opt} if and only if $\nabla g(x^\star)=0$. 

A direct computation of the gradient of $g$ implies that  $\nabla g(x) = 0$ if and only if $\nabla f(x) = g(x) \nabla\{\|x\|_p\}$. Suppose $x\succeq 0$. We have  $\nabla\{\|x\|_p\} = (\|x\|_p)^{1-p} x^{p-1}$  and 
$\nabla f(x) = F(x)$ with 
$$
F(x) = \Diag(x)^{q-1}  B \Xi(B^\top x^q)^{\frac 1 q-1}\, .
$$
Thus,   $\nabla g(x)= 0$ with $x\succeq 0$ if and only if $F(x)=\lambda x^{p-1}$, for some $\lambda >0$.

For a $p>1$, let $p^*=p/(p-1)$ denote its H\"older conjugate such that $1/p+1/q=1$. Define 
$$
H(x) = \frac{F(x)^{p^*-1}}{\|F(x)^{p^*-1}\|_p} = \frac{F(x)^{\frac{1}{p-1}}}{\|F(x)^{\frac{1}{p-1}}\|_p }=\Big(\frac{F(x)}{\|F(x)\|_{p^*}}\Big)^{\frac{1}{p-1}}\, . 
$$
From $F(x)=\lambda x^{p-1}$, we have that $\nabla g(x)=0$ for $x\succeq 0$ if and only if $H(x) = x/\|x\|_p$ and,  for a point   such that $\|x\|_p=1$, if and only if $H(x) = x$. 
We will show next that $H$ has a unique fixed point such that $x^\star \succ 0$ and $\|x^\star \|_p=1$ and that the sequence $x^{(k+1)}=H(x^{(k)})$ converges to $x^\star$ for any $x^{(0)}$ with positive entries. This implies that $f(x)$ has only one critical point $x^\star$ with positive entries and of unit $p$-norm. Note that the sequence $x^{(k)}$ coincides with the iterations generated by scheme in the statement of Theorem \ref{thm:main}, thus showing the convergence of $x^{(k)}$ and its convergence rate will conclude the proof. 



To this end, we will show that $H$ is a contraction with respect to the Thompson distance, defined as  
$$
d(x,y) = \|\ln x - \ln y\|_\infty\, ,
$$
for any two $x,y\succ 0$. Precisely, we will show that for any two points $x,y \in S_+ =\{x:x\succ 0, \|x\|_p=1\}$ we have $H(x),H(y) \in S_+$ and 
$$
d(H(x),H(y)) \leq \left| \frac{q-1}{p-1}\right| \, d(x,y) \, .
$$
Since $S_+$ is complete with respect to $d$ (see e.g.\ \cite{lemmens_nussbaum}), the Banach fixed point theorem will then directly imply the thesis.


In order to prove the contraction bound above we need a number of careful computations. First, using the mean value theorem we have 
\begin{align}
\begin{aligned}\label{eq:mean_value_thm}
    d(H(x),H(y)) &= \|\ln(H(x))-\ln(H(y))\|_\infty \\
    &= \|\ln(H(\exp(\ln(x))))-\ln(H(\exp(\ln(y))))\|_\infty\\
    &\leq \sup_{z \in \Omega(x,y)}\|\mathcal F(z) (\ln(x)-\ln(y))\|_\infty 
    \leq d(x,y) \sup_{z \in \Omega(x,y)}\|\mathcal F(z)\|_\infty,
\end{aligned}
\end{align}
where $\Omega(x,y)$ is the line segment joining $x$ and $y$, and $\mathcal F(z)$ denotes the Frech\'et derivative of the map $\ln \circ H \circ \exp$ evaluated at $z$. A direct computation using the chain rule shows that 
$$
\mathcal F(z) = \Diag(H(e^z))^{-1}\D H(e^z)\Diag(e^z) = \Diag(H(\tilde z))^{-1}\D H(\tilde z)\Diag(\tilde z),
$$
where $\tilde z\in \tilde \Omega:=\exp(\Omega(x,y))$ and where $\D$ denotes the Jacobian operator. 

Using the chain rule several times we compute
$$
\D H(z) = \frac 1 {p-1} \left(\frac{F(z)}{\|F(z)\|_{p^*}}\right)^{\frac 1 p -1} \left\{\frac{\D  F(z)}{\| F(z)\|_p}-\frac{ F(z)(\nabla \| F(z)\|_p)^T \D  F(z)}{\| F(z)\|_p^2}\right\} 
$$
which implies that
\begin{align*}
\begin{aligned}
|\Diag(H(z))^{-1} \D H(z)| &= |\Diag( F(z))^{\frac 1 {1-p}} \| F(z)\|_{p^*}^{\frac 1 p-1} \D H(z)|\\
&=\frac 1 {|p-1|}\left|\Diag( F(z))^{-1} \D  F(z) - \frac{\mathbb 1 \, (\nabla \| F(z)\|_{p^*})^\top \D F(z)}{\|F(z)\|_{p^*} }\right|\, .
\end{aligned}
\end{align*}
For $z\succ 0$, we have $\nabla \| F(z)\|_{p^*} = \| F(z)\|^{1-p^*}_{p^*} F(z)^{p^*-1}$. Thus
\begin{align*}
\begin{aligned}
|p-1|\Big(|\Diag(H(z))^{-1} \D H(z)|\Big)_{ij} &= \left| \frac{\D  F(z)_{ij} }{ F(z)_i} - \sum_i \frac{\| F(z)\|_{p^*}^{1-p^*} F(z)_i^{p^*-1}\D  F(z)_{ij} } {\big(\sum_k  F(z)_k^{p^*} \big)^{1/p^*} }\right| \\
&= \left| \frac{\D  F(z)_{ij} }{ F(z)_i} - \sum_i \frac{\| F(z)\|_{p_*}^{1-p^*} F(z)_i^{p^*}  }  {\big(\sum_k  F(z)_k^{p^*} \big)^{1/p^*} } \cdot \frac{\D  F(z)_{ij}}{F(z)_i}\right|\\
&=: |C_{ij} - \sum_i \gamma_i C_{ij}|,
\end{aligned}
\end{align*}
where $C_{ij} = \D  F(z)_{ij} /  F(z)_i$ and $\gamma_i =  \| F(z)\|_p^{-1} F(z)_i^p$. Notice that $C_{ij}$ and $\gamma_i$ are all nonnegative numbers (for all $i,j$) and that $\sum_i \gamma_i=1$. Therefore, 
\[
|p-1|\Big|\Big(\Diag H(x)^{-1} |\D H(x)|\Big)_{ij}\Big|=|C_{ij} - \sum_i \gamma_i C_{ij}| \leq \max_{i=1,\dots,n} C_{ij} 
\]
As a consequence, if $i_\star$ is the index such that $C_{i_\star j} = \max_{i=1,\dots,n} C_{ij}>0$,  for any $z\succeq 0$, we obtain 
\begin{align}\label{eq:lipschitz_bound}
\begin{aligned}
\Big\|\Diag H(z)^{-1}|\D H(z)|\Diag(z)\Big\|_\infty 
&= \max_i \sum_j \left|\Big(\Diag(H(z))^{-1} |\D H(z)|\Big)_{ij}z_j \right| \\
&\leq \frac 1 {|p-1|} \sum_j \big| C_{i_\star j}z_j\big| 
= \frac 1 {|p-1|} \big|  \sum_j C_{i_\star j}z_j\big|  \\
&=  
\frac 1 {|p-1|} \left|\frac{\big(\D  F(z)z\big)_{i_\star}}{ F(z)_{i_\star}}\right| \, .
\end{aligned}
\end{align}
Now, recall that for a generic node $i$ we have $F(z)_i = z_i^{q-1}\sum_{e:i\in e}\xi(e)(\sum_{k\in e}z_k^q)^{1/q-1}$. Thus, if $\delta_{ij}$ denotes the Kronecker delta ($\delta_{ij}=1$ if $i=j$,  $\delta_{ij}=0$ otherwise), we have \begin{align*}
    \frac {\partial F(z)_i}{\partial z_j} &= (q-1)\Big\{\delta_{ij}z_i^{q-2}\sum_{e:i\in e}\xi(e)\big(\sum_{k\in e}z_k^q\big)^{\frac 1 q -1} - z_i^{q-1}\sum_{e: i,j\in e}\xi(e)\big(\sum_{k\in e}z_k^q\big)^{\frac 1 q -2}z_j^{q-1}\Big\}\\
\end{align*}
so that 
\begin{align}
\begin{aligned}
     \frac 1 {|q-1|}\left|\frac{\big(\D  F(z)z\big)_{i}}{ F(z)_{i}}\right| &= \frac 1 {|q-1|}\left|\frac{1}{ F(z)_{i} } \sum_j \frac {\partial F(z)_i}{\partial z_j} z_j \right| \\
    &\leq  \frac 1 {F(z)_i} \left\{\left| 
     z_i^{q-1}\sum_{e:i\in e}\xi(e)\big(\sum_{k\in e}z_k^q\big)^{\frac 1 q -1} -
    z_i^{q-1}\sum_{e: i\in e}\xi(e)\big(\sum_{k\in e}z_k^q\big)^{\frac 1 q -2}z_i^{q} \right| \right. \\
    &+\left.\left|    z_i^{q-1}\sum_{j\neq i} \sum_{e: i,j\in e}\xi(e)\big(\sum_{k\in e}z_k^q\big)^{\frac 1 q -2}z_j^{q}\right|\right\} \\
    & = \frac 1 {F(z)_i} \left\{ \left|z_i^{q-1}\sum_{e:i\in e}\xi(e)\big(\sum_{k\in e}z_k^q\big)^{\frac 1 q -1} \frac{\sum_{k\in e\setminus\{i\} }z_i^q} {\sum_{k\in e} z_k^q}\right|\right. \\
    &+ \left.\left|z_i^{q-1} \sum_{j\neq i} \sum_{e: i,j\in e}\xi(e)\big(\sum_{k\in e}z_k^q\big)^{\frac 1 q -1} \frac{z_j^{q}}{\sum_{k\in e}z_k^q} \right|\right\} \\
    &\leq \!\frac 1 {F(z)_i} \left\{z_i^{q-1}\sum_{e:i\in e}\xi(e)\big(\sum_{k\in e}z_k^q\big)^{\frac 1 q -1} + z_i^{q-1} \sum_{j\neq i} \sum_{e: i,j\in e}\xi(e)\big(\sum_{k\in e}z_k^q\big)^{\frac 1 q -1}\right\} \! =\! 1
    \end{aligned}   \label{eq:lip_2}
\end{align}
where the final inequality follows from the fact that, for $z\succ 0$, we have 
\[
0\leq \frac{\sum_{k\in e\setminus\{i\} }z_i^q} {\sum_{k\in e} z_k^q} \leq 1 \qquad \text{and} \qquad 0\leq \frac{z_j^{q}}{\sum_{k\in e}z_k^q} \leq 1 \, .
\]
Finally, combining  \eqref{eq:lipschitz_bound} with \eqref{eq:lip_2}, we obtain that
$$
\sup_{z \in \Omega(x,y) }\|\mathcal F(z)\|_\infty  = \sup_{z\in \tilde \Omega} \Big\|\Diag H(z)^{-1}|\D H(z)|\Diag(z)\Big\|_\infty \leq |q-1|/|p-1|
$$
concluding the proof. 



\end{document}